\documentclass[a4paper,12pt]{article}
\usepackage{graphicx} 

\usepackage{tikz}

\usepackage{amsmath,amsthm,amscd,amssymb,eucal,mathrsfs}
\usepackage{amsfonts}
\usepackage{latexsym}
\usepackage{dsfont}
\usepackage[all]{xy}

\newtheorem{thm}{Theorem}[section]

\newtheorem{definition}[thm]{Definition}

\newtheorem{Lemma}[thm]{Lemma}

\newcommand{\mathset}[1]{{\left\{#1\right\}}}
\newcommand{\absolute}[1]{\left\lvert#1\right\rvert}

\title{On the Local Ultrametricity of Finite Metric Data
}
\author{Patrick Erik Bradley
\\
Karlsruhe Institute of Technology
\\
Institute of Photogrammetry and Remote Sensing
\\
Englerstr.\ 6
\\
76131 Karlsruhe
\\
Germany
}
\date{\today}

\begin{document}

\maketitle

\newpage
\begin{abstract}
New local ultrametricity measures for finite metric data are proposed through  the viewpoint that their Vietoris-Rips corners are samples from $p$-adic Mumford curves endowed with a Radon measure coming from a regular differential $1$-form. This is experimentally applied to the iris dataset.
\end{abstract}



\section{Introduction}

Ultrametricity is appealing for many reasons, and in particular  the simplicty of tree structures encoded in ultrametric spaces seems attractive to data analysts.
Because of this, they would like to see how close to an ultrametric space a given data set is in order to extract something meaningful out of a hierarchical classification of the data.
With this in mind, ultrametricity indices have been proposed, e.g.\ by \cite{RTV1986} or \cite{Murtagh2004}. F.\ 
Murtagh observed experimentally that samples which are sparse and random in hypercubes
become more and more ultrametric as dimension increases, using his  ultrametricity index
\cite{Murtagh2004}.
Explanations for this are given in \cite{Brad-topoUltra} and \cite{Zubarev2014}. 
Also, ultrametricity can be related to topological data analysis \cite{Brad-topoUltra2}, and a corresponding ultrametricity index has a logistic behaviour \cite{Brad-topoUltra3}. This index relies on the Vietoris-Rips complex developped in \cite{Vietoris1927}, and which is important in studying the persistent homology of data. Cf.\ e.g.\ \cite{Zomorodian2010} for a fast construction of the Vietoris-Rips complex.
\newline

The $p$-adic numbers, having an inherent regular hierarchical structure, provide a framework for analysing hierarchical data, and thus $p$-adic encoding methods were devised, \cite{Murtagh2016} or \cite{Brad-MumfDendro}, either in order to bring them closer to ultrametricity or to apply $p$-adic methods to their already existing hierarchical structure.
This leads to the applicability of $p$-adic analysis outlined e.g.\ in \cite{VVZ1994} to the investigation of data.
\newline

The scope of this article is to introduce new measures for \emph{local} ultrametricity, arguing that the clusters appearing as connected components of the Vietoris-Rips graphs are likely to be more ultrametric than the whole dataset, which can be seen in an example case taken from the well-known iris dataset. Whether or not this argument is generally valid or not, the viewpoint induced by this approach leads to the idea that data can be seen as being sampled from Mumford curves. These are $p$-adic compact algebraic manifolds of dimension $1$. Locally, they are holed discs in the $p$-adic number field, on which there is a natural Haar measure. However, the irregular tree structure of the local data leads to a more natural Radon measure coming from an algebraic regular differential $1$-form on the Mumford curve, as constructed in \cite{IndexTopo_p}. There, the subdominant ultrametric associated with a finite metric space is used, which can be calculated with the method of \cite{RTV1986}.
In fact, any ultrametric can be used to approximate the finite metric dataset, i.e.\ any hierarchical classification method can be used in order to obtain an ultrametric in this approach.
\newline

Mumford curves are objects studied in $p$-adic algebraic and rigid geometry, and are extensively covered in \cite{GvP1980} and 
\cite{FvP2004}. What is needed from this relatively deep theory is, however, only the fact that they are algebraic and have an underlying $1$-dimensional compact $p$-adic manifold structure which allows for regular differential $1$-forms $\omega$, which are in fact algebraic. Locally, they are of the form
\[
\omega(x)=f(x)\,dx
\]
with an analytic $p$-adic-valued function $f$ defined on the local piece $U$,
and that these give rise to Radon measures on the Mumford curve outside the zeros of $\omega$.

\section{Finite locally ultrametric spaces}

After defining local ultrametrics in the following subsection,  and local $p$-adic encodings of data via tree embeddings, new invariants of a finite metric space are defined via the Vietoris-Rips graphs and  associating Mumford curves and Radon measures to local pieces in the last subsection of this section.

\subsection{Local ultrametrics}
Let $X$ be a finite set with a metric $d$ on it.
Fix $\epsilon>0$, and let
$\Gamma_\epsilon$ be the associated Vietoris-Rips graph with vertex set $X$. Let 
\[
d_\epsilon\colon X\times X\to\mathds{R}_{\ge0}
\]
be the partial function which on each connected component $C$ of $\Gamma_\epsilon$ is an  ultrametric dominated by $d$.
One can use for $d_\epsilon$ e.g.\ the subdominant
corresponding to the distance on $X$ restricted to $C\times C$. But any other ultrametric  dominated by $d$ can also be used.
Certain hierarchical clustering methods provide such an ultrametric, among which single-linkage clustering yields the subdominant ultrametric.
\newline

Let
$\mathcal{C}(\Gamma_\epsilon)$ be the set of connected components of $\Gamma_\epsilon$. Define
a distance $d'_\epsilon$ on $\mathcal{C}(\Gamma_\epsilon)$ 
as
\[
d'_\epsilon(C,C')
=\min\mathset{\epsilon'\mid\epsilon'\ge\epsilon\colon
\exists\;\text{an edge in $\Gamma_{\epsilon'}$ connecting $C$ and $C'$}}
\]
whenever $C\neq C'$.
Then define the function
\[
\delta_\epsilon
\colon X\times X\to\mathds{R}_{\ge0},\;(x,y)\mapsto
\begin{cases}
d_\epsilon(x,y),&\exists\, C\in\mathcal{C}(\Gamma_\epsilon)\colon x,y\in C
\\
d'_\epsilon(C(x),C(y)),&
C(x)\neq C(y)
\end{cases}
\]
where $C(z)\in\mathcal{C}(\Gamma_\epsilon)$ is the connected component containing $z\in X$.
Clearly, $\delta_\epsilon$ is a distance on $X$.

\begin{definition}
The distance $\delta_\epsilon$ is called a \emph{local ultrametric} on $X$. The pair $(X,\delta_\epsilon)$ is called a \emph{locally ultrametric space}.
\end{definition}

A criterion for ultrametricity in terms of the Vietoris-Rips graphs is given in \cite[Lem.\ 2.2]{Brad-topoUltra}: the
connected components of the Vietoris-Rips graphs are always cliques iff dataset is ultrametric. The  subdominant ultrametric can also be described in terms of the Vietoris-Rips graphs, cf.\ \cite[Prop.\ 5.2]{Brad-topoUltra}.

\subsection{Local $p$-adic encodings}

In \cite[\S 3.3]{IndexTopo_p} a Radon measure on a compact open subset of $\mathds{Q}_p$ is constructed from a finite ultrametric space. Here, an embedding of the corresponding ultrametric tree into the Bruhat-Tits tree of a suitable $p$-adic number field necessary for that method is constructed in a more precise manner. This produces a $p$-adic data encoding, as already observed in \cite{Brad-MumfDendro}.
\newline

Let $C\in\mathcal{C}(\Gamma_\epsilon)$ be given, and view $(C,d_\epsilon)$ as
an independent ultrametric space for the moment.
The set $\mathcal{B}(C)$ of all non-trivial balls on $C$ is a finite poset with precisely one top element $C$, and in fact is a tree. Let
\[
\rho\colon\mathcal{B}(C)\to\mathds{R}_{>0},
B\mapsto\text{radius of $B$}
\]
whose image $R(C)=\rho(\mathcal{B}(C))$ is a finite ordered set of real  numbers. Order this set with a function
\[
\varphi\colon R(X)\to\mathds{N}
\]
in decreasing order with consecutive natural numbers beginning in $0$. 
Fix a prime number $p$, and let
\[
m=\max\varphi
\]
and assign to each $c\in C$ a distinct disc $a_c+p^{(m+1)}\mathds{Z}_p$ inside the ring $\mathds{Z}_p$ of $p$-adic integers inside the field of $p$-adic numbers $\mathds{Q}_p$, where $p$ is bounded from below by the maximal number of children in any ultrametric tree of $(C,\delta_\epsilon)$ for any $C\in \mathcal{C}(\Gamma_\epsilon)$ plus the number of elements in $\mathcal{C}(\Gamma_\epsilon)$.  Assume thereby that  all discs in $\mathds{Z}_p$ have equal radius
\[
p^{-(m+1)}
\]
for this assignment.
The condition about $p$  
enables an embedding of any spanning tree of the graph $\Gamma_\epsilon$ into the Bruhat-Tits tree for $\mathds{Q}_p$. The latter tree is explained e.g.\ in \cite{Brad-dendrofam}.
\newline

The ultrametric diffusion considered in \cite{IndexTopo_p} necessitated the replacement of the $p$-adic Haar measure on the compact open set obtained by such an embedding as above with a Radon measure $\nu_C$ which ensures that the volume of any disc corresponding to a vertex $B\in\mathcal{B}(C)$ in the ultrametric tree for $(c,d_\epsilon)$ is equally distributed among the child vertices of $B$, cf.\ \cite[Lem.\ 3.8]{IndexTopo_p}, where such a Radon measure is constructed on the subset
\[
\Omega_C=\bigsqcup\limits_{c\in C}\left(a_c+p^{m+1}\mathds{Z}_p\right)
\]
of $\mathds{Q}_p$.

\begin{definition}
The measure $\nu_C$ is called the \emph{equity} measure on $\Omega_C$ induced by $(C,d_\epsilon)$.
\end{definition}

Now, \cite[Lem.\ 3.9]{IndexTopo_p} shows that $\nu_C$ is of the form
\[
\nu(x)=\phi(\absolute{f_C(x)})\,\absolute{dx}_p
\]
where $\absolute{dx}_p$ is the Haar measure on $\mathds{Q}_p$, $f_C\in \mathds{Q}_p[X]$ a polynomial nowhere vanishing on $\Omega_C$,
and  $\phi\colon p^{\mathds{Z}}\to\mathds{R}_{>0}$ a strictly increasing function. The proof of \cite[Lem.\ 3.9]{IndexTopo_p}
uses $p$-adic polynomial interpolation.

\subsection{Invariants of a finite metric space}

Let $\delta\ge\epsilon$.
Define the graph $\Gamma_\epsilon^\delta$ whose vertex set is $\mathcal{C}(\Gamma_\epsilon)$, and its edges are pairs $(C,C')$ with $C\neq C'$ and $d(C,C')\le\delta$.
We call this the \emph{coarse $\epsilon$-$\delta$ graph} of $(X,d)$.
\newline

For $x\in X$ define
\[
C_\epsilon^\delta(x)=\text{the connected comp.\ of $\Gamma_\epsilon^\delta$ containing $C(x)\in\mathcal{C}(\Gamma_\epsilon)$}
\]
where $C(x)\in\mathcal{C}(\Gamma_\epsilon)$ is the connected component of $\Gamma_\epsilon$ containing $x\in X$.
The \emph{local $\epsilon$-$\delta$-genus} is the function
\[
g_\epsilon^\delta\colon X\to\mathds{N},\;x\mapsto b_1(C_\epsilon^\delta(x))
\]
where $C_\epsilon^\delta$ is viewed as a subgraph of $\Gamma_\epsilon^\delta$.

\begin{Lemma}
It holds true that
\[
0\le g_\epsilon^\delta(x)\le\frac12\absolute{C_\epsilon^\delta(x)}^2-\frac32\absolute{C_\epsilon^\delta(x)}+1
\]
for $\delta\ge\epsilon>0$.
$(X,d)$ is ultrametric, if and only if for all $\delta\ge\epsilon>0$, the right hand side is an equality.
\end{Lemma}

\begin{proof}
The $\epsilon$-$\delta$-genus is non-negative and is maximal, when $C_\epsilon^\delta(x)$ is a complete graph, in which case the first Betti number equals the right hand side of the asserted inequality. The last statment now follows immediately from \cite[Lem.\ 2.2]{Brad-topoUltra} applied to the quotient metric on the set $\mathcal{C}(\Gamma_\epsilon)$ induced by $d$.
\end{proof}

Such a connected component $C=C_\epsilon^\delta(x)$ for $x\in X$ can be viewed as a coarse graph structure on the connected components of $\Gamma_\epsilon$. 
Now, replacing, as in the previous subsection, the  distance $d$, restricted to 
each element of $\mathcal{C}(\Gamma_\epsilon)$ which is contained in $C$,
with an ultrametric, leads to a local tree structure on $C$.  
Now, the previous subsection
tells us that locally, there is a Radon measure $\nu(x)$ coming from a $p$-adic differential $1$-form $\omega_\epsilon$ which is algebraic. The local pieces can now be ``glued'' to a covering of the $p$-adic points of a
Mumford curve $\mathcal{C}_\epsilon$ minus the zeros of the regular algebraic differential $1$-form
$\omega_\epsilon$, whose genus equals the first Betti number of $C$. Notice that the gluing process takes place beyond the mere points whose coordinates are in $\mathds{Q}_p$, in the category of $p$-adic rigid analytic spaces, as laid out e.g.\ in \cite[Ch.\ 5]{FvP2004}. The method from the previous subsection fills a gap of \cite[\S 3.3]{IndexTopo_p} by explicitly constructing the equity measure on the Mumford curve, which according to \cite[Lem.\ 3.9]{IndexTopo_p} comes from an algebraic differential $1$-form $\omega_\epsilon$. 
\newline

Hence, apart from the local $\epsilon$-$\delta$-genus, there is also the equity measure $\absolute{\omega_\epsilon}_p$ on each connected component of $\Gamma_\epsilon^\delta$ as a further set of invariant of the dataset $(X,d)$. 
As another invariant, we suggest also the minimal value $\delta$ for which $\Gamma_\epsilon^\delta$  is connected, together with the now global $\epsilon$-$\delta$-genus and equity measure for this $\delta$.

\section{Experiments}

The iris dataset was investigated. For $\epsilon\ge1.65$, the Vietoris-Rips graph $\Gamma_\epsilon$ becomes connected. Figure \ref{vrgs}
shows Vietoris-Rips graphs for four selected values of $\epsilon$. Taking as reference $\epsilon=0.64$ leads to five clusters (i.e.\ connected components of $\Gamma_\epsilon$) having distance matrix as in Table \ref{clusterdist064}. 
Reference $\epsilon=0.7$ leads to four clusters with distance matrix as in Table \ref{clusterdist07}.
The clusters $C_2,C_3$ of the first graph merge to cluster $C$ of the second graph. Figure \ref{MumfordGraphs} shows the two quotient graphs $\Gamma_{\epsilon}^\delta$ for these values of $\epsilon$, and for the corresponding minimal $\delta>\epsilon$, for which the associated Mumford curve is connected.
In the first case, the 
genus is one, and in the second case it is zero.
\newline

Our choice of an ultrametricity is biased by \cite[\S 3.3]{Murtagh2004} who introduced his ultrametricity index also because of the chaining effect problem of the ultrametricity index from \cite{RTV1986}. 
The values of the Murtagh ultrametricity index of the not too small clusters are given in Table \ref{MurtaghIndex} and indicate that they are more ultrametric than the whole iris dataset whose Murtagh index is given as $0.0162$ in \cite[\S 3.3]{Murtagh2004}.
\newline

The method of \cite[\S 3.3]{Murtagh2004} of calculating the Murtagh ultrametricity index consists in 
 counting almost ultrametric triangles as follows:

\begin{enumerate}
\item Randomly sample the coordinates for triples of three distinct points.

\item Check for possible  degenerate triangles and exclude these.

\item The cosine of the angle facing a side of
length $x$ is: 
\[
\frac{y^2 + z^2 - x^2}
{2yz} 
\]
where $y,z$ are the other side lengths.
\item For the two other angles seek an angular difference of at most  2 degrees (0.03490656 radians).
\end{enumerate}
Murtagh's ultrametricity index is then the fraction $\alpha$ of such almost ultrametric triangles.
\newline

\begin{table}[ht]
\[
\begin{array}{|c||c|c|c|c|c|}\hline
\epsilon=0.64&C_1&C_2&C_3&C_4&C_5\\\hline\hline
C_1&0&2.08&1.64&3.14&5.46\\\hline
C_2&2.08&0&0.648&0.735&0.819\\\hline
C_3&1.64&0.648&0&1.32&4.59\\\hline
C_4&3.14&0.735&1.32&0&3.79\\\hline
C_5&5.46&0.819&4.59&3.79&0
\\\hline
\end{array}
\]
\caption{Cluster distances at $\epsilon=0.64$.}\label{clusterdist064}
\end{table}

\begin{table}[ht]
\[
\begin{array}{|c||c|c|c|c|}\hline
\epsilon=0.7&C_1&C&C_4&C_5\\\hline\hline
C_1&0&1.64&3.14&5.46\\\hline
C&1.64&0&0.735&0.819\\\hline
C_4&3.14&0.835&0&3.79\\\hline
C_5&5.46&0.819&3.79&0\\\hline
\end{array}
\]
\caption{Cluster distances at $\epsilon=0.7$.}\label{clusterdist07}
\end{table}

\begin{table}[h]
\[
\begin{array}{|c||c|c|c|}\hline
\text{cluster}&C_1&C_2&C
\\\hline
\text{Murtagh index}&0.026&0.11&0.12
\\\hline
\end{array}
\]
\caption{The Murtagh ultrametricity index values for the larger clusters in the iris dataset.}\label{MurtaghIndex}
\end{table}

The $3$-adic Radon measure leading to the equity measure on one cluster is now computed as an exemplary study of the iris dataset.
The subdominant ultrametric of cluster $C_3$ in $\Gamma_{0.64}$ can be depicted as the dendrogram w.r.t.\ single-linkage clustering, and is shown in Figure \ref{dendro_c3}.
In the case of $p=3$,
assign now to each point of $C_3$ a $3$-adic ball of radius $3^{-2}$, centred in
\[
0,\quad 1,\quad
1+3^2,\quad 
1+3^2+3^3,
\]
a set which has the $3$-adic tree structure of Figure \ref{dendro_c3}.
The equity measure
leads to the assignment
\[
0\mapsto \frac12,\quad
1\mapsto\frac14,\quad
1+3^2\mapsto\frac18,\quad
1+3^2+3^3\mapsto\frac18
\]
leading to the $3$-adic interpolation problem
\[
f(0)=1,\quad f(1)=3,\quad
f(1+3^2)=3^2,
\quad
f(1+3^2+3^3)=3^2
\]
with the solution
\[
f(X)=\frac{31}{9990}X^3-\frac{1683}{9990}X^2+\frac{10811}{4995}X+1
\in\mathds{Q}_3[X]
\]
for the new measure
\[
\absolute{\omega(x)}_3=\absolute{f(x)}_3\,\absolute{dx}_3
\]
approximating the equity measure for $x\in C_3$ with $\epsilon=0.64$.

\begin{figure}[ht]
\begin{center}
\begin{tabular}{|c|c|}\hline
\includegraphics[scale=.3]{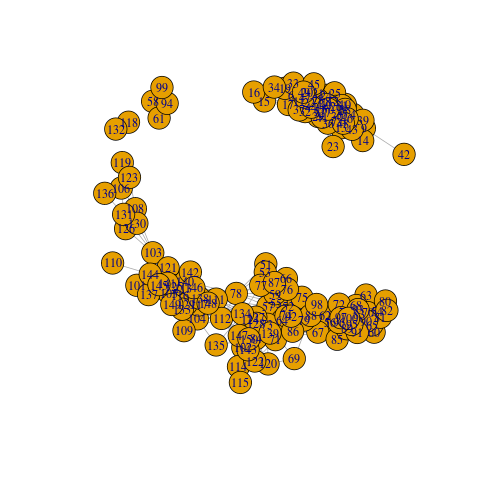}&\includegraphics[scale=.3]{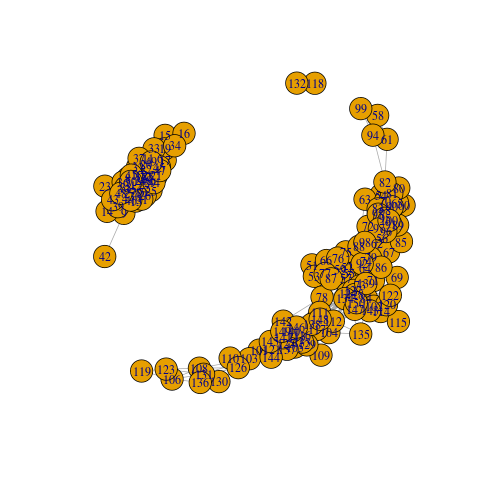}
\\\hline
\includegraphics[scale=.3]{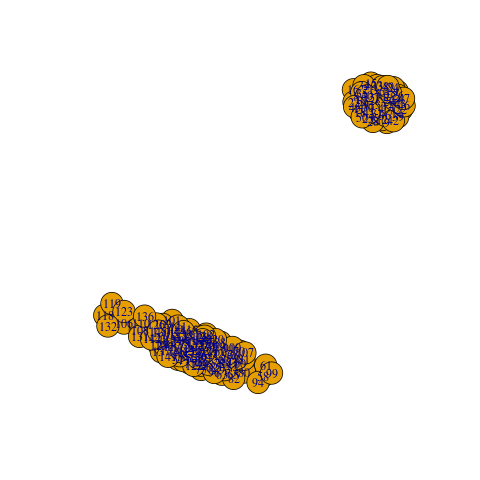}
&\includegraphics[scale=.3]{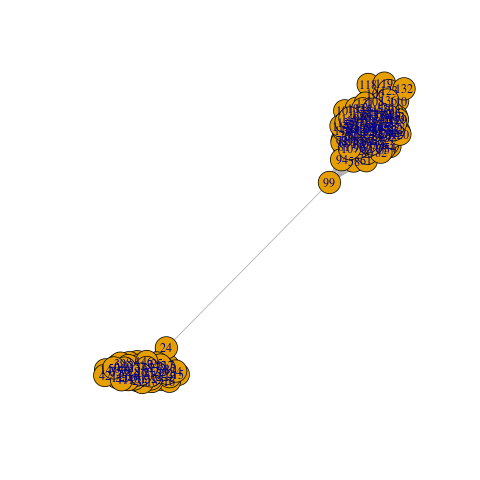}
\\\hline
\end{tabular}
\end{center}
\caption{Vietoris-Rips graphs $\Gamma_\epsilon$ for $\epsilon=0.64$ (top left), $\epsilon=0.7$ (top right), $\epsilon=1.640$ (bottom left), $\epsilon=1.650$ (bottom right). Only the non-singleton connected components are shown.}\label{vrgs}
\end{figure}

\begin{figure}[ht]
\[
\xymatrix@C=40pt{
C_1\ar@{-}[drr]_(.25){\!1.64}&C_2\ar@{-}[dr]^{0.648}\ar@{-}[dd]^(.55){0.735}\ar@{-}[ddl]^(.6){\!0.819}
\\
&&C_3\ar@{-}[dl]^{1.32}
\\
C_5&C_4
}
\qquad\qquad
\xymatrix@=50pt{
C_1\ar@{-}[r]^{1.64}&C\ar@{-}[d]^{0.735}\ar@{-}[dl]^(.52){\!0.819}
\\
C_5&C_4
}
\]
\caption{The coarse graphs $\Gamma_\epsilon^\delta$ for $(\epsilon,\delta)=(0.64,1.64)$ (left), and $(\epsilon,\delta)=(0.7,1.64)$ (right).
}\label{MumfordGraphs}
\end{figure}
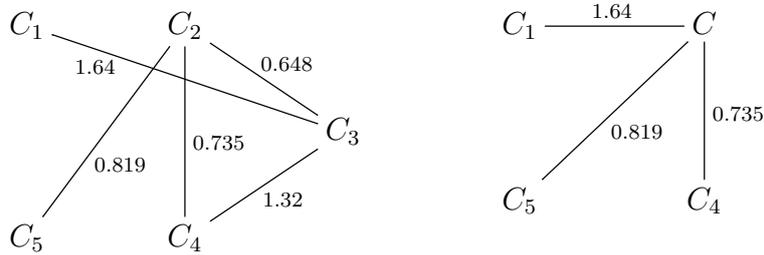

\begin{figure}[ht]
\begin{center}
\includegraphics[scale=.5]{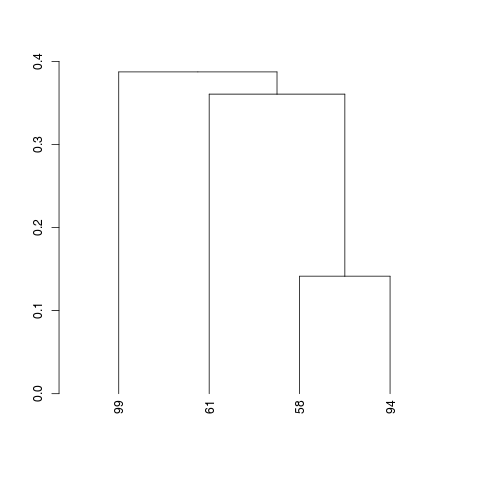}
\caption{Single-linkage dendrogram of the cluster $C_3$ in $\Gamma_{0.64}$.}
\end{center}\label{dendro_c3}
\end{figure}

\section{Conclusion}

In the theoretical part of this work, a local ultrametric is associated with a finite metric space (given the meaning of a dataset) via its Vietoris-Rips graph. The equity measure was defined on the local ultrametric parts of the space which, by the result of recent previous work, can be seen as coming from a differential $1$-form on a $p$-adic Mumford curve.
These can be viewed as compact algebraic $p$-adic manifolds, and this approach leads to new invariants for finite metric data by taking a double filtration with $\epsilon$- and $\delta$-balls in the finite metric.
In the experimental part of this work, it is conceivable
from the findings with the iris dataset, that local ultrametricity could increase within clusters of the Vietoris-Rips graphs in comparison with the ultrametricity of the whole dataset. 
This resembles \emph{Simpson's paradox} in statistics. 
Even if this may not be always the case for general datasets,  this nevertheless provides a means for classifying different types of datasets.
In the iris data case, example Vietoris-Rips graphs were taken and some values of the new invariants calculated. The findings suggest that the double filtration approach can reveal more inherent  topological properties of data in their ultrametric approximation via $p$-adic encoding. This suggests that in the future, a hierarchical or $p$-adic version of manifold learning could emerge from further investigations in this direction, which is appealing because of its potential for a reduced computational complexity.

\section*{Acknowledgements}


\section*{Data Availability}
This research has used publicly available data only. The $R$ code can be made available upon request.

\bibliographystyle{plain}
\bibliography{biblio}

\end{document}